\documentclass[10pt,conference]{IEEEtran}

\usepackage{graphicx,epic,eepic,epsfig,amsmath,latexsym,amssymb,verbatim,subfigure,color}
\usepackage{theorem}

\newtheorem{definition}{Definition}

\newtheorem{lemma}[definition]{Lemma}

\newtheorem{theorem}[definition]{Theorem}
\newtheorem{corollary}[definition]{Corollary}

\def\squareforqed{\hbox{\rlap{$\sqcap$}$\sqcup$}}
\def\qed{\ifmmode\squareforqed\else{\unskip\nobreak\hfil
\penalty50\hskip1em\null\nobreak\hfil\squareforqed
\parfillskip=0pt\finalhyphendemerits=0\endgraf}\fi}
\def\endenv{\ifmmode\;\else{\unskip\nobreak\hfil
\penalty50\hskip1em\null\nobreak\hfil\;
\parfillskip=0pt\finalhyphendemerits=0\endgraf}\fi}

\mathchardef\ordinarycolon\mathcode`\:
\mathcode`\:=\string"8000
\def\vcentcolon{\mathrel{\mathop\ordinarycolon}}
\begingroup \catcode`\:=\active
  \lowercase{\endgroup
  \let :\vcentcolon
  }

\newcommand{\nc}{\newcommand}
\nc{\rnc}{\renewcommand}
\nc{\beq}{\begin{equation}}
\nc{\eeq}{{\end{equation}}}
\nc{\beqa}{\begin{eqnarray}}
\nc{\eeqa}{\end{eqnarray}}
\nc{\lbar}[1]{\overline{#1}}
\nc{\bra}[1]{\langle#1|}
\nc{\ket}[1]{|#1\rangle}
\nc{\ketbra}[2]{|#1\rangle\!\langle#2|}
\nc{\braket}[2]{\langle#1|#2\rangle}
\nc{\proj}[1]{| #1\rangle\!\langle #1 |}
\nc{\avg}[1]{\langle#1\rangle}
\nc{\Rank}{\operatorname{Rank}}
\nc{\smfrac}[2]{\mbox{$\frac{#1}{#2}$}}
\nc{\Tr}{\operatorname{Tr}}
\nc{\tr}{\operatorname{Tr}}
\nc{\id}{\operatorname{id}}
\nc{\1}{\openone}
\nc{\ox}{\otimes}
\nc{\dg}{\dagger}
\nc{\dn}{\downarrow}
\nc{\cA}{{\cal A}}
\nc{\cB}{{\cal B}}
\nc{\cC}{{\cal C}}
\nc{\cD}{{\cal D}}
\nc{\cE}{{\cal E}}
\nc{\cF}{{\cal F}}
\nc{\cG}{{\cal G}}
\nc{\cH}{{\cal H}}
\nc{\cI}{{\cal I}}
\nc{\cJ}{{\cal J}}
\nc{\cK}{{\cal K}}
\nc{\cL}{{\cal L}}
\nc{\cM}{{\cal M}}
\nc{\cN}{{\cal N}}
\nc{\cO}{{\cal O}}
\nc{\cP}{{\cal P}}
\nc{\cR}{{\cal R}}
\nc{\cS}{{\cal S}}
\nc{\cT}{{\cal T}}
\nc{\cX}{{\cal X}}
\nc{\cY}{{\cal Y}}
\nc{\cZ}{{\cal Z}}
\nc{\supp}{{\operatorname{supp}}}
\nc{\var}{\operatorname{var}}
\nc{\rar}{\rightarrow}
\nc{\lrar}{\longrightarrow}
\nc{\polylog}{\operatorname{polylog}}

\def\a{\alpha}

\def\ph{\varphi}

\def\D{\Delta}

\nc{\RR}{{{\mathbb R}}}
\nc{\CC}{{{\mathbb C}}}
\nc{\FF}{{{\mathbb F}}}
\nc{\NN}{{{\mathbb N}}}
\nc{\ZZ}{{{\mathbb Z}}}
\nc{\PP}{{{\mathbb P}}}
\nc{\QQ}{{{\mathbb Q}}}
\nc{\UU}{{{\mathbb U}}}
\nc{\EE}{{{\mathbb E}}}
\nc{\Icoh}{{I^c}}
\nc{\Qca}{{Q_{\rm ss}}}
\nc{\Qcaa}{{Q^{(1)}_{\rm ss}}}
\nc{\Dcaa}{{D^{(1)}_{{\rm ss}\rightarrow}}}
\nc{\Dca}{{D_{{\rm ss}\rightarrow}}}

\nc{\be}{\begin{equation}}
\nc{\ee}{{\end{equation}}}
\nc{\bea}{\begin{eqnarray}}
\nc{\eea}{\end{eqnarray}}
\nc{\<}{\langle}
\rnc{\>}{\rangle}
\nc{\Hom}[2]{\mbox{Hom}(\CC^{#1},\CC^{#2})}
\nc{\rU}{\mbox{U}}

\begin{document}

\author{
\authorblockN{Graeme Smith}
\authorblockA{IBM TJ Watson Research Center\\
1101 Kitchawan Road\\
 Yorktown NY 10598\\
graemesm@us.ibm.com}
\and
\authorblockN{John A. Smolin}
\authorblockA{IBM TJ Watson Research Center\\
1101 Kitchawan Road\\
 Yorktown NY 10598\\
smolin@watson.ibm.com}
}

\title{Additive extensions of a quantum channel}

\maketitle

\begin{abstract}
We study extensions of a quantum channel whose one-way capacities are
described by a single-letter formula.  This provides a simple
technique for generating powerful upper bounds on the capacities of a
general quantum channel.  We apply this technique to two qubit
channels of particular interest---the depolarizing channel and the
channel with independent phase and amplitude noise.  Our study of the
latter demonstrates that the key rate of BB84 with one-way
post-processing and quantum bit error rate $q$ cannot exceed $H(1/2 -
2q(1-q)) - H(2q(1-q))$.
\end{abstract}

\section{Introduction}

Perhaps the central problem of information theory is finding the rate at which
information can be transmitted through a noisy channel.  Indeed, Shannon
created the field with his 1948 paper \cite{Shannon48} showing that
the capacity of a noisy channel is equal to the maximum mutual
information over all input distributions to {\em a single use} of the
channel, even though the encoding needs, in general, to use an
asymptotically large number of channel uses.

However, it has long been known that the apparent quantum
generalization of the mutual information, namely the {\em coherent
information}, does not yield a single-letter formula for the quantum
information capacity $Q$ \cite{SS96,DSS98}.  Similarly, though the
private capacity of a classical broadcast channel is known, and given
by a single-letter formula \cite{CK78}, the private classical capacity
of a quantum channel is not known.  

The quantum capacity is given by \cite{Lloyd97,D03,BNS98}:
\begin{equation}
Q=\lim_{n\rightarrow\infty} \frac{1}{n} {\rm max}_{\phi_n} 
I^{\rm c}\left({\cal N}^{\ox n}, {\phi_n}\right) 
\end{equation}
where
\begin{equation}
I^{\rm c}(\cN,\phi)=I^{\rm c}
\left(I{\otimes}{\cal N}(\proj{\phi^{AB}})\right)\ .
\end{equation}
Here $\ket{\phi^{AB}}$ is a purification of $\phi$ and $I^{\rm
c}(\rho_{AB})=S(\rho_B){-}S(\rho_{AB})$ with $S(\rho){=}-{\rm
Tr}(\rho\log \rho)$.
The private capacity is given by \cite{D03}
\begin{equation}
C_p(\cN) = \lim_{n \rightarrow \infty} \frac{1}{n} C^{(1)}_p(\cN^{\ox n})
\end{equation}
where 
\begin{equation}
C^{(1)}_p(\cN) \equiv \sup_{\{p_x, \ket{\ph_x}\}, X\rightarrow T}
\left(I(T;B)_\omega - I(T;E)_\omega \right),
\end{equation}
with $\omega_{ABE} = \sum_{x,t}p(t|x)p(x)\proj{t}_A \ox
U_{\cN}\proj{\ph_x}U_{\cN}^\dg$ and $U_\cN$ an isometric extension of
$\cN$ (i.e., $\cN(\rho) = \Tr_E U_\cN \rho U_\cN^\dg$).  The mutual information
is defined, as usual, according to $I(T;B)_{\omega_{TB}} = S(T)_{\omega_T} + S(B)_{\omega_B} - S(BT)_{\omega_{BT}}$, where
we have used subscripts on the states to indicate which system they live on (e.g., $\omega_B = \Tr_T \omega_{BT}$), and used the
notation $S(B)_{\omega_B} = S(\omega_B)$.  When it is clear which state we are referring to, we will omit the subscript on the entropy.

Since the formulas for these capacities involve maximizations over
ever growing numbers of channel uses, we cannot evaluate them at all.
This unsatisfying situation is a reflection of our lack of
understanding of how to choose asymptotically good codes.  Our best
understanding is presented in \cite{SmithSmo0506}.

Fortunately, we {\em can} evaluate capacity for some channels---degradable
ones \cite{DS03}.  In general, channels for which the coherent information is
additive, {\em i.e.} 
\begin{equation}
Q(\cN) = Q^{(1)}(\cN) \equiv \max_\phi I^c(\cal N,\phi),
\end{equation}
of which degradable channels are an example, are much easier
to deal with than arbitrary channels.  Once we have an understanding
of additive channels, we can use them to bound the capacities of
other channels \cite{SSW06} \cite{S07}.  Here we improve upon that
work and develop new tighter and much simpler upper bounds.

First we will define the concepts of additive and degradable 
extensions to a channel and prove they have single-letter formulas for their
capacities.  We then use a particularly simple class of degradable extensions,
which we call 'flagged extensions' to bound the quantum and 
private capacities.  We also show that the best known previous techniques
are special cases of our new bound.  Finally, we bound the key rate of 
BB84 quantum key distribution \cite{BB84} for a channel with bit error
rate $q$ by $H(1/2 - 2q(1-q))-H(2q(1-q))$.

\section{Additive and Degradable Extensions}

\begin{definition}
We call $\cT$ an {\em additive extension} of a quantum channel $\cN$
if there is a second channel $\cR$ such that $\cN = \cR \circ \cT$ and
$Q(\cT) = Q^{(1)}(\cT)$.\footnote{Of course, we could also define additive extensions which have $C_p^{(1)}(\cT) = C_p(\cT)$ in order to 
find upper bounds on $C_p(\cN)$.  However, since the only channels we know with $C^{(1)}_p(\cT) = C_p(\cT)$ also have 
$C^{(1)}_p(\cT) = Q^{(1)}(\cT)$, we will not pursue this approach here.}
\end{definition}
A particularly nice type of additive extension is one which satisfies
the following definition:
\begin{definition}
A channel $\cN$, with isometric extension $U: A \rightarrow BE$ is
called {\em degradable} if there is a {\em degrading map} $\cD$ such
that $\cD \circ \cN = \widehat{\cN}$, where $\widehat{\cN}(\rho) =
\Tr_B U\rho U^\dg$.  $\widehat{\cN}$ is called the {\em complementary
channel} of $\cN$.
\end{definition}

We call an additive extension of a quantum channel that is degradable
a {\em degradable extension}.  Degradable extensions have the
additional property that their coherent information is an upper bound
for the private classical capacity as well as the quantum capacity.  

Our main tool will be the following simple theorem, which bounds the capacity of a quantum channel in terms of 
the capacity of its additive extensions.

\begin{theorem}\label{Thm:Main}
The quantum capacity of a channel $\cN$ satisfies
\begin{equation}
Q(\cN) \leq  Q^{(1)}(\cT),
\end{equation}
for all additive extensions, $\cT$, of $\cN$.  Furthermore, if $\cT$ is degradable, the private classical capacity
of $\cN$ satisfies
\begin{equation}
C_p(\cN) \leq Q^{(1)}(\cT).
\end{equation}
\end{theorem}
\begin{proof}
$Q^{(1)}(\cT)=Q(\cT)$ and $Q(\cN)\le Q(\cT)$ follows immediately from
the fact that $\cN$ can be obtained from $\cT$ by apply $\cR$.  If $\cT$ is degradable, 
it was shown in \cite{S07} that 
\begin{equation}
C_p(\cT) = Q^{(1)}(\cT),
\end{equation}
so that we hvae $C_p(\cN)\leq C_p(\cT)= Q^{(1)}(\cT)$.
\end{proof}

\section{Known upper bounds}

In this section we show that the two strongest techniques for upper
bounding the capacities of a quantum channel are encompassed by our
approach.  The first technique, established in \cite{SSW06,WP06} and
best for low noise levels, is to decompose the channel into a convex
combination of degradable channels.  The second, first studied in
\cite{BDSW96,BDEFMS98,Cerf00}, is a no-cloning type argument that can
sometimes be used to show that a very noisy channel has zero capacity.

\subsection{Convex combinations of degradable channels}

\begin{lemma}\label{Lem:Flagged}
Suppose we have 
\begin{equation}
\cN = \sum_i p_i   \cN_i,
\end{equation}

where $\cN_i$ is degradable with degrading map $\cD_i$.  Then 
\begin{equation}
\cT = \sum_i p_i \cN_i \ox \proj{i}
\end{equation}
is a degradable extension of $\cN$, and

\begin{equation}
Q(\cN) \leq \sum_i p_iQ^{(1)}(\cN_i).
\end{equation}
We will call $\cT$ a {\em flagged} degradable extension on $\cN$ since
$i$ keeps track of which $\cN_i$ actually occurred in the 
decomposition of $\cN$.
\end{lemma}
\begin{proof}
First, let $\cR$ be the partial trace on the flagging system so that
$\cN = \cR \circ \cT$.
To see that $\cT$ is degradable, note that the complementary channel of $\cT$ is
\begin{equation}
\widehat{\cT} = \sum_i \widehat{\cN_i} \ox \proj{i}
\end{equation}
and that letting
\begin{equation}
\cD = \sum_i \cD_i \ox \proj{i},
\end{equation}
where $\cD_i \cT_i = \widehat{\cT_i}$, we have $\cD\circ \cT = \widehat{\cT}$.

Finally, letting $\phi$ be the optimal input state for $\cT$, we find
\begin{eqnarray}
\nonumber Q^{(1)}(\cT) &=& S\left(\sum_i p_i \cN_i(\phi)\otimes \proj{i}\right) \\
\nonumber & & - S\left(\sum_i p_i \widehat{\cN}_i(\phi)\otimes \proj{i}\right)\\
\nonumber & = & \sum_i p_i \left( S\left(\cN_i(\phi) \right) - S\left(\widehat{\cN}_i(\phi)\right)\right)\\
\nonumber  & \leq & \sum_i p_i Q^{(1)}(\cN_i),
\end{eqnarray}
so that by Theorem 2 the result follows.
\end{proof}

\subsection{No-cloning bounds}

We next show that no-cloning bounds \cite{BDEFMS98,Cerf00} are a special case.

Suppose $\cN$ is antidegradable, meaning there is a channel $\cD$ such
that $\cD \circ \widehat{\cN} = \cN$.  In this case, we can define a
zero-capacity degradable extension of $\cN$ as follows.  Let $\cN$ have isometric extension
$U:A\rightarrow BE$, $d = \max(d_B,d_E)$, and $F_1$ and $F_2$ be
$d$-dimensional spaces with $B\subset F_1$ and $E\subset F_2$.  Then
define isometry $V:A\rightarrow F_1F_2C_1C_2$ as
\begin{equation}
V\ket{\phi} = \nonumber \frac{1}{\sqrt{2}}U\ket{\phi}\ket{01}_{C_1C_2}+\frac{1}{\sqrt{2}}({\rm SWAP_{F_1F_2}}U\ket{\phi})\ket{10}_{C_1C_2}
\end{equation}

This gives a degradable extension of $\cN$, $\cT(\rho) =
\Tr_{F_2C_2}V\rho V^\dg$, which can be degraded to $\cN$.

\subsection{Convexity of bounds}
We now show that if we have upper bounds for the capacity
of two channels, both obtained from a degradable extension, the convex
combination of the bounds is an upper bound for the capacity of the
corresponding convex combination of the channels.  More concretely,
suppose $\cT_0$ and $\cT_1$ are degradable extensions of $\cN_0$ and
$\cN_1$, respectively.  Then,
\begin{equation}
\cT = p\cT_0\ox\proj{0}+ (1-p)\cT_1\ox\proj{1}
\end{equation}
is a degradable extension of $\cN = p\cN_0 + (1-p)\cN_1$, and satisfies
\begin{eqnarray}
Q^{(1)}(\cT)&=& p I^c(\cT_0,\phi) + (1-p)I^c(\cT_1,\phi)\\
& \leq & p Q^{(1)}(\cT_0) + (1-p)Q^{(1)}(\cT_1).
\end{eqnarray}

\section{Bounds on Specific Channels}
In this section we will evaluate explicit upper bounds on the private classical and quantum capacities of 
the depolarizing channel and Pauli channels with independent amplitude and phase noise (which we also call ``the BB84 channel'', because
of its relevance for BB84).  In each case, 
we will use a flagged degradable extension of the channel of interest, based on the convex decomposition 
into degradable channels used in \cite{SSW06}\cite{S07}.  The advantage we obtain over this previous work is, essentially,
due to the fact that our upper bound involves a maximization of the average coherent informations of the elements of 
our decomposition, all with respect to the same state.  In \cite{SSW06}\cite{S07}, the corresponding bound is the average
of the individual maxima, allowing {\em different} reference states for each channel in the decomposition, which 
generally leads to a weaker bound.

Throughout this section, we will use the following special property of coherent information for
degradable channels, which was first proved in \cite{YDH05}.  It will
assist in the evaluation of coherent informations for specific
degradable extensions below.

\begin{lemma}
Let $\cN$ be degradable.  Then
\begin{equation}\nonumber
pI^c(\cN,\phi_0) + (1-p)I^c(\cN,\phi_1) \leq I^c(\cN,p\phi_0+(1-p)\phi_1).
\end{equation}
In other words, $I^c(\cN,\phi)$ is concave as a function of $\phi$.
\end{lemma}
\begin{proof}
Writing out the entropies involved explicitly, what we would like 
to prove is that
\begin{eqnarray}
\nonumber 
& & pS\left(\cN(\phi_0)\right)+(1-p)S\left(\cN(\phi_0)\right)\\
& & -S\Bigl(p\cN(\phi_0)+(1-p)\cN(\phi_1)\Bigr)\nonumber\\
&\leq &  \nonumber\\
\nonumber  
& & pS\left(\widehat{\cN}(\phi_0)\right)+(1-p)S\left(\widehat{\cN}(\phi_0)\right) \\
& & -S\left(p\widehat{\cN}(\phi_0)+(1-p)\widehat{\cN}(\phi_1)\right),
\end{eqnarray}
which, letting $U$ be the isometric extension of $\cN$ and 
\begin{equation}
\rho_{VBE} = p \proj{0}_V \ox U\phi_0U^\dg + (1-p)\proj{1}_V\ox U\phi_1 U^\dg,
\end{equation}
is equivalent to 
\begin{equation}
H(V|B)_{\rho_{VB}} \leq H(V|E)_{\rho_{VE}}.
\end{equation}
Noting that $H(U|B)$ is nondecreasing under operations on $B$ (which is a simple consequence of the strong subadditivity of 
quantum entropy), and there is a $\cD$ that maps $B$ to $E$ completes the proof.
\end{proof}

We will also have use for $\cN_{(u,v)}$, the most general
degradable qubit channel (up to unitary operations on the input 
and output) \cite{CRS08}.  $\cN_{(u,v)}$ has Kraus operators
\begin{eqnarray}
A_+ &=& \left( \begin{matrix} \cos(\frac{1}{2}(v-u)) & 0\\ 0 &
 \cos(\frac{1}{2}(v+u)) \end{matrix}\right) \\ A_- &=& \left(
 \begin{matrix} 0 & \sin(\frac{1}{2}(v+u))\\ \sin(\frac{1}{2}(v-u)) &
 0 \end{matrix}\right).
\end{eqnarray}
In \cite{WP06}, $\cN_{(u,v)}$ was shown to be degradable when 
$|\sin{v}| \leq |\cos{u}|$.

\subsection{Depolarizing Channel}
The depolarizing channel of error probability $p$ is given by
\begin{equation}\nonumber
\cN_p(\rho) = (1-p)\rho + \frac{p}{3} X\rho X + \frac{p}{3} Y\rho Y + \frac{p}{3} Z\rho Z \ .
\end{equation}
It is particularly nice to study since it has the property that for
any unitary $U$
\begin{equation}
\cN_p(U\rho U^\dag)= U \cN_p(\rho) U^\dag\ .
\end{equation}

The following theorem, together with the subsequent corollary, provides the strongest upper bounds to date on the
capacity of the depolarizing channel.  We provide a proof of the theorem below, after establishing two essential lemmas.
\begin{theorem}
The capacity of the depolarizing channel with error probability $p$ satisfies
\begin{equation}
Q(\cN_p) \leq {\rm co}\left [ \D (p), 1-4p\right],
\end{equation}
where
\begin{equation}\nonumber
\D(p)=\min H\left[\frac{1}{2}[1{+}\sin u\sin v]\right]{-}H\left[\frac{1}{2}[1{+}\cos u\cos v]\right],
\end{equation}
with the minimization over $(u,v)$ such that $\cos^2(u/2)\cos^2(v/2) = 1-p$, and ${\rm co}[f_1(p),f_2(p)\ldots f_n(p)]$ denotes the maximal convex function that is less than or equal to all $f_i(p)\ , i=1 \ldots n$.
\label{TH6}
\end{theorem}

\begin{corollary}
\begin{equation}\nonumber
Q(\cN_p) \leq {\rm co}\left[\!1-H(p),H(\frac{1-\gamma(p)}{2})-H(\frac{\gamma(p)}{2}), 1-4p\!\right],
\end{equation}
where $\gamma(p) =  4 \sqrt{1-p}(1 - \sqrt{1-p})$.  \label{C7}
\end{corollary}
\begin{proof}
Corollary \ref{C7} follows from the theorem by noting that the first
two terms inside the square brackets are special cases of $\Delta$ for 
values of $(u,v)$ corresponding to amplitude damping and to dephasing
channels, and using the fact that the true minimum is always bounded
by particular cases.
\end{proof}

To establish Theorem \ref{TH6} we will use the following flagged
degradable extension of the depolarizing channel:
\begin{equation} \label{Eq:UVDegradableExtension}
\cT^{\rm dep}_{(u,v)}(\rho) = \frac{1}{|\cC|}\sum_{c\in \cC} 
c^\dg\cN_{(u,v)}(c\rho c^\dg)c \ox \proj{c},
\end{equation}
where $\cal C$ is the set of unitaries which map
$\{I,X,Y,Z\}\rightarrow \{I,X,Y,Z\}$ under conjugation (the Clifford group).

\begin{lemma}\label{Lem:EvalQ1}
\begin{equation}\nonumber
Q^{(1)}(\cT^{\rm dep}_{(u,v)}){=}H\left[\frac{1}{2}[1{+}\sin u \sin v]\right]{-}H\left[\frac{1}{2}[1{+}\cos u \cos v]\right]
\end{equation}
\end{lemma}
\begin{proof}
The main step is to show that the coherent information of $\cT^{\rm dep}_{(u,v)}$ is maximized by the maximally mixed state.  To see this, first note that for any $\phi$, we have
\begin{eqnarray}
& & (X\ox I) \cT^{\rm dep}_{(u,v)}(X\phi X) (X\ox I)  \\ 
\nonumber && = \frac{1}{|{\cal C}|}\sum_{c\in \cC} Xc^\dg\cN_{(u,v)}(cX\phi Xc^\dg)cX \ox \proj{c}\\
\nonumber && = \frac{1}{|{\cal C}|}\sum_{c\in \cC} Xc^\dg\cN_{(u,v)}(cX\phi Xc^\dg)cX \ox V\proj{cX}V^\dg \\
\nonumber && =  (I \ox V) \cT^{\rm dep}_{(u,v)}(\phi) (I\ox V^\dg),
\end{eqnarray}
where we have chosen unitary $V$ such that $V\ket{cX} = \ket{c}$.  Since
\begin{equation}
{\widehat{\cT}^{\rm dep}_{(u,v)}}(\phi) =  \frac{1}{|{\cal C}|}\sum_{c\in \cC} c^\dg\widehat{\cN}_{(u,v)}(c\phi c^\dg)c \ox \proj{c},
\end{equation}
by an identical argument we also get that
\begin{equation}
{\widehat{\cT}^{\rm dep}_{(u,v)}}(X\phi X) = (X\ox V) {\widehat{\cT}^{\rm dep}_{(u,v)}}(\phi )(X\ox V^\dg).
\end{equation}
Thus, we have
\begin{eqnarray}
S\left(\cT^{\rm dep}_{(u,v)}(\phi)\right) &=& S\left(\cT^{\rm dep}_{(u,v)}(X\phi X) \right)\\
S\left(\widehat{\cT}^{\rm dep}_{(u,v)}(\phi)\right) &=& S\left(\widehat{\cT}^{\rm dep}_{(u,v)}(X\phi X) \right),
\end{eqnarray}
so that $I^c(\cT^{\rm dep}_{(u,v)},\phi) = I^c(\cT^{\rm dep}_{(u,v)},X\phi X)$, and similarly for $Y$ and $Z$.  Using the concavity of $I^c(\cT^{\rm dep}_{(u,v)},\phi)$ in $\phi$, this gives us
\begin{eqnarray}
&& I^c(\cT^{\rm dep}_{(u,v)},\phi)  \\
& = &\frac{1}{4}\sum_{P\in \cP}I^c(\cT^{\rm dep}_{(u,v)},P\phi P^\dg)\\
& \leq & I^c\left(\cT^{\rm dep}_{(u,v)},\frac{1}{4}\sum_{P\in \cP}P\phi P^\dg\right)\\
& = & I^c\left(\cT^{\rm dep}_{(u,v)},\frac{1}{2} I\right),
\end{eqnarray}
where we have let $\cP = \{ I,X,Y,Z\}$ denote the Pauli matrices.  
This shows that the maximum coherent information is achieved for the reference state $I/2$, where its value is
\begin{eqnarray}
\nonumber I^c\left(\cT^{\rm dep}_{(u,v)},\frac{1}{2} I\right) & = & S(\cN_{(u,v)}(I/2))- S(\widehat{\cN}_{(u,v)}(I/2))\\
\nonumber  & = & H\left(\frac{1}{2}[1 + \sin u \sin v]\right)\\
 & &  - H\left(\frac{1}{2}[1 + \cos u \cos v]\right).
\end{eqnarray}
\end{proof}

The following lemma shows that $\cT^{\rm dep}_{(u,v)}$ can be degraded to a depolarizing channel, and computes the
error probability of that channel as a function of $u$ and $v$.

\begin{lemma}\label{Lem:UVErrorProb}

\begin{equation}
\Tr_2 \cT^{\rm dep}_{(u,v)}(\rho) = \cN_{p(u,v)}(\rho),
\end{equation}
where $p(u,v) = 1 - \cos^2(u/2)\cos^2(v/2)$.

\end{lemma}

\begin{proof}
First, note that for any channel $\cN$, it was shown in \cite{BDSW96} that 

\begin{equation}
\widetilde{\cN}(\rho ) = \frac{1}{|{\cal C}|}\sum_{c\in {\cal C}} c^\dg \cN(c \rho c^\dg)c
\end{equation}

is a depolarizing channel with the same entanglement fidelity as $\cN$.  In other words, 

\begin{equation}
\bra{\phi^+} (I \ox \cN) (\proj{\phi^+}) \ket{\phi^+} = \bra{\phi^+} (I \ox \widetilde{\cN}) (\proj{\phi^+}) \ket{\phi^+},
\end{equation}
where $\ket{\phi^+} = (\frac{1}{\sqrt{2}})(\ket{00}+\ket{11})$.  As a result, letting $\cN_{p(u,v)} = \widetilde{\cN}_{(u,v)}$ define $p(u,v)$,
and using the fact that $1-p(u,v) = \bra{\phi^+} (I \ox \cN_{p(u,v)})(\proj{\phi^+})\ket{\phi^+}$, we have
\begin{eqnarray}
1-p(u,v)&= & \bra{\phi^+} (I \ox \cN_{(u,v)}) (\proj{\phi^+})\ket{\phi^+} \nonumber \\
& = & \bra{\phi^+} I \ox A_+ \proj{\phi^+} I \ox A_+^\dg \ket{\phi^+}\nonumber\\
\nonumber  & &  +  \bra{\phi^+} I \ox A_- \proj{\phi^+} I \ox A_-^\dg \ket{\phi^+}.
\end{eqnarray}
Using the fact that 
\begin{equation}\nonumber
 \bra{\phi^+} I \ox A \proj{\phi^+} I \ox A^\dg \ket{\phi^+} = \left( \frac{1}{2}\Tr A\right)\left(\frac{1}{2}\Tr A^\dg\right)
 \end{equation}
and $A_-$ is traceless then gives
\begin{eqnarray}
1-p(u,v) &=& \left( \frac{1}{2}\Tr A_+\right)^2\\
\nonumber & =  & \left( \frac{1}{2}\left(\cos((v-u)/2) + \cos((v+u)/2) \right)\right)^2\\
\nonumber & = & \cos^2(v/2)\cos^2(u/2).
\end{eqnarray}
\end{proof}

\begin{proof}{\em (of Theorem \ref{TH6})}
Let $1 - p(u,v) = \cos^2(u/2)\cos^2(v/2)$, and $\cT^{\rm dep}_{(u,v)}$ be the 
channel described in Eq.~(\ref{Eq:UVDegradableExtension}).
Then, by Lemma \ref{Lem:Flagged}, $\cT^{\rm dep}_{(u,v)}(\rho)$ is degradable, and by  
Lemma \ref{Lem:UVErrorProb} tracing over the flag system degrades $\cT^{\rm dep}_{(u,v)}$ to $\cN_{p(u,v)}$.  Thus, 
$\cT^{\rm dep}_{(u,v)}$ is a degradable extension of $\cN_{p(u,v)}$.  As a result, by Theorem \ref{Thm:Main} we have
\begin{equation}
Q(\cN_{p(u,v)}) \leq Q^{(1)}(\cT^{\rm dep}_{(u,v)}),
\end{equation}
whereas by Lemma \ref{Lem:EvalQ1}  
\begin{equation}\nonumber
Q^{(1)}(\cT^{\rm dep}_{(u,v)}){=}H\left[\frac{1}{2}[1{+}\sin u \sin v]\right]{-}H\left[\frac{1}{2}[1{+} \cos u \cos v]\right].
\end{equation}

Furthermore, it was shown in \cite{BDEFMS98} that $\cN_p$ becomes anti-degradable when $p=1/4$, so that $Q(\cN_{1/4}) = 0$.

Since both of these bounds are the result of arguments via a degradable extension, their convex hull is also an upper bound for 
the capacity of $\cN_{p}$, which completes the proof.
\end{proof}

\begin{figure}[htbp]
\includegraphics[width=3.00in]{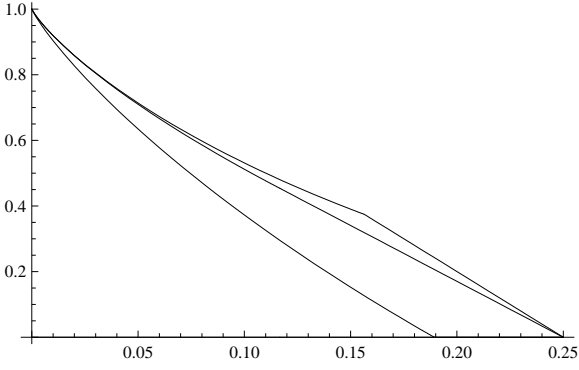}
\caption{Bounds on the quantum capacity of the depolarizing channel with error probability $p$.  The horizontal
axis is the error probability, and the vertical axis is the rate.  The lowest
line is the achivable rate using hashing \protect\cite{BDSW96}.  The top line
is the minimum of $1-H(p)$ \protect\cite{Rains-PPT} and 
$1-4p$ \protect\cite{BDSW96}.  The middle line is our new bound from
Corollary \protect\ref{C7}.  While the bound from  Theorem \protect\ref{TH6}
is tighter everywhere, it, the bound plotted, and the bound from
\protect\cite{SSW06} all look essentially identical on this scale.
}
\end{figure}

\subsection{The BB84 Channel}
Bennett-Brassard quantum key distribution \cite{BB84} is the most
widely studied and practically applied form of quantum cryptography.
A simple bound on the achievable key rate is therefore quite useful.
Here we evaluate the coherent information sharable though a degradable
extension of the BB84 channel, which also bounds the secret key rate
of this protocol.

Define the following degradable channel
\begin{equation}\nonumber
\cT^{\rm BB84}_q = \frac{1}{2}\cN^{\rm ad}_{\gamma(q)}(\rho )\ox \proj{0} +
\frac{1}{2}Y\cN^{\rm ad}_{\gamma(q)}(Y\rho Y)Y\ox \proj{1},
\end{equation}
with $\gamma(q) = 4q(1-q)$, and $\cN^{\rm ad}_{\gamma}$ the amplitude damping channel with Kraus
operators 
\begin{eqnarray}
A_0 &=& \left(\begin{matrix} 1 & 0 \\ 0 & \sqrt{1-\gamma}\end{matrix}\right)\\
A_1 &=& \left(\begin{matrix} 0 & \sqrt{\gamma} \\ 0 & 0\end{matrix}\right).
\end{eqnarray}
The channel $\cN^{\rm ad}_\gamma$ is a special case of $\cN_{(u,v)}$ with $u=v=\cos^{-1}(\sqrt{1-\gamma})$, which is degradable as long
as $\gamma \leq 1/2$ \cite{GF04}.

By tracing out the flag system, we find that
\begin{eqnarray}
\nonumber\Tr_2 \cT^{\rm BB84}_{q}(\rho)&=&(1-q)^2\rho + q(1-q)X\rho X\\
\nonumber && + q^2Z\rho Z + q(1-q) Y\rho Y.
\end{eqnarray}
With suitable unitary rotations on the input and output, this can be transformed to 
\begin{eqnarray}
\cN^{\rm BB84}_{q}(\rho)&\equiv& (1-q)^2\rho + q(1-q)X\rho X\\
 & & + q(1-q)Z\rho Z + q^2 Y\rho Y.
\end{eqnarray}
The channel $\cN^{\rm BB84}_q$ is such that its private classical capacity is equal to the maximal achievable 
key rate in BB84 with one-way postprocessing and quantum bit error rate $q$.  Since $\cT^{\rm BB84}_q$ is a degradable
extension of an equivalent channel, its coherent information provides an upper bound on the key rate of this protocol.

The coherent information of $\cT^{\rm BB84}_q$ can readily be evaluated, giving the
following upper bound.

\begin{lemma}
\begin{equation}
C_p(\cN^{\rm BB84}_q)\leq H\left(\frac{1}{2} -2q(1-q)\right) - H(2q(1-q))
\end{equation}
\end{lemma}
\begin{proof}
This follows, via Theorem \ref{Thm:Main}, from the fact that $\cT^{\rm BB84}_q$ is a degradable extension of
$\cN_q^{\rm BB84}$ together with Lemma \ref{Lem:EvalBB84}, which evaluates this extension's coherent information.
\end{proof}

\begin{lemma}\label{Lem:EvalBB84}
\begin{equation}
Q^{(1)}(\cT^{\rm BB84}_q) = H\left(\frac{1}{2} -2q(1-q)\right) - H(2q(1-q))
\end{equation}
\end{lemma}
\begin{proof}
We would like to evaluate
\begin{equation}
\max_{\phi} I^c(\cT^{\rm BB84}_q,\phi).
\end{equation}
For any $\phi$, we have 

\begin{equation}
I^c(\cT^{\rm BB84}_q,\phi) = I^c(\cT^{\rm BB84}_q,Y\phi Y)
\end{equation}
so that, using the concavity of $I^c$ in $\phi$ for degradable channels, we have
\begin{equation}
I^c\left(\cT^{\rm BB84}_q,\phi\right) 
\leq I^c\left(\cT^{\rm BB84}_q,\frac{1}{2}\phi + \frac{1}{2}Y\phi Y\right)
\end{equation}
As a result, we may take the optimal $\phi$ to be of the form 
\begin{equation}
\phi_\a = \frac{1}{2}I + \frac{\a}{2}Y.
\end{equation}
Now, 
\begin{equation}
\nonumber
I^c(\cT^{\rm BB84}_q,\phi_\a) = S\left(\cN_{4q(1-q)} (\phi_\a)\right) - S\left(\cN_{1-4q(1-q)} (\phi_{\a}) \right),
\end{equation}
which can be written more explicitly as
\begin{eqnarray}\label{Eq:ExplicitBB84}
H \left( \frac{1}{2}\left(1 - \sqrt{\gamma^2 + \a^2 (1-\gamma)}\right) \right) \\
- H \left( \frac{1}{2}\left(1 - \sqrt{(1-\gamma)^2 + \a^2 \gamma}\right) \right),\nonumber
\end{eqnarray}
from which we see that  $I^c(\cT^{\rm BB84}_q,\phi_\a) = I^c(\cT^{\rm BB84}_q,\phi_{-\a})$.  Using the concavity of $I^c$ again, we see 
\begin{eqnarray}
I^{\rm}(\cT^{\rm BB84}_q) & = & I^c(\cT^{\rm BB84}_q,\frac{1}{2} I)\\
& \geq &  I^c(\cT^{\rm BB84}_q,\phi_\a),
\end{eqnarray}
and evaluating Eq.~(\ref{Eq:ExplicitBB84}) for $\a=0$ gives the result.
\end{proof}

\section{Additive Extensions and Symmetric Assistance}
There is an entertaining connection to the capacity of a
quantum channel with symmetric assistance \cite{SSW06}.  
We first briefly summarize the main finding of \cite{SSW06}, using
slightly more streamlined notation.  Letting $\cH^\prime = {\rm
span}\{\ket{(i,j)}\}_{i<j\in \ZZ^+}$ and $\cH = {\rm
span}\{\ket{i}\}_{i \in \ZZ^+}$, and defining the partial isometry
$V:\cH^\prime \rightarrow \cH \otimes \cH$ to act as $V\ket{(i,j)} =
\frac{1}{\sqrt{2}}(\ket{i}\ket{j} - \ket{j}\ket{i})$, we call the
channel $\cA: \cB(\cH^\prime)\rightarrow \cB(\cH)$ that acts as
$\cA(\rho) = \Tr_2V\rho V^\dg$ the symmetric assistance channel.  It
was shown in \cite{SSW06} that for any channel $\cN$, the quantum
capacity of $\cN$ given free access to $\cA$ is given by
\begin{equation}
Q_{ss}(\cN)  =  Q^{(1)}(\cN\ox \cA),
\end{equation}
where the single-letter nature of this expression comes from the non-obvious 
fact that 
\begin{equation}
Q^{(1)}(\cN\ox \cM\ox\cA) =  Q^{(1)}(\cN \ox \cA)+Q^{(1)}(\cM\ox\cA)\ .
\end{equation}
We note that since $\cA$ is an infinite dimensional operator, 
$\cA$  is equally valuable for assistance to $\cA^{\otimes n}$  
Note also that $Q(\cA)=Q^{(1)}(\cA)=Q(\cA^{\ox n})=Q^{(1)}(\cA^{\ox n})=0$ 
by a no-cloning argument.  Then
\begin{eqnarray}
Q(\cN \ox \cA)=\lim_{n\rightarrow \infty} \frac{1}{n}
Q^{(1)}(\cN^{\ox n}\ox \cA^{\ox n})\\
=\lim_{n\rightarrow \infty} \frac{1}{n}\left(Q^{(1)}(\cN^{\ox n}\ox \cA) + 
Q^{(1)}(\cA^{\ox (n-1)} \ox \cA)\right)\\
=\lim_{n\rightarrow\infty}\frac{1}{n}(nQ^{(1)}(\cN\ox\cA))=
Q^{(1)}(\cN\ox\cA)\ .
\end{eqnarray}
Therefore $\cN\ox\cA$ is an additive extension of $\cN$ for any $\cN$.

JAS thanks ARO contract DAAD19-01-C-0056.

\begin{thebibliography}{10}
\bibitem{S07}
G.~Smith, ``The private classical capacity with a symmetric side channel and
  its application to quantum cryptography,'' ar{X}iv:0708:1021.
\end{thebibliography}

\end{document}